\documentclass[11pt]{article}
\usepackage[margin=1.25in]{geometry}
\usepackage{amssymb,amsfonts,amsmath,amsthm,amscd,dsfont,mathrsfs,bbold,pifont}
\usepackage{blkarray}
\usepackage{graphicx,float,psfrag,epsfig,color}
\usepackage{comment}
\usepackage{subcaption}

\usepackage{algorithm}
\usepackage[noend]{algpseudocode}

\makeatletter
\def\BState{\State\hskip-\ALG@thistlm}
\makeatother

	\usepackage{microtype}
	\usepackage[pdftex,pagebackref=true,colorlinks]{hyperref}
	\hypersetup{linkcolor=[rgb]{.7,0,0}}
	\hypersetup{citecolor=[rgb]{0,.7,0}}
	\hypersetup{urlcolor=[rgb]{.7,0,.7}}

\DeclareMathAlphabet{\mathpzc}{OT1}{pzc}{m}{it}

\theoremstyle{plain}
\newtheorem{Def}{Definition}
\newtheorem{exa}{Example}
\newtheorem{prop}{Proposition}
\newtheorem{lemma}{Lemma}
\newtheorem{corol}{Corollary}
\newtheorem{claim}{Claim}
\theoremstyle{remark}
\newtheorem{remark}{Remark}
\theoremstyle{discussion}
\theoremstyle{plain}
\newtheorem{conj}{Conjecture}

\def\F{{\mathbb F}}

\def\1{{\mathds 1}}

\newcommand{\mZ}{\mathbb{Z}}

\newcommand{\e}{\varepsilon}
\newcommand{\rank}{\mathrm{rank}}

\begin{document}

%
%

\title{Linear Boolean classification, coding and \\``the critical problem''\footnote{Part of this work was presented at ISIT 2014 \cite{isit_annulus}.}
}

\author{Emmanuel Abbe\thanks{Program in Applied and Computational Mathematics, and EE department, Princeton University, Princeton, USA, \texttt{eabbe@princeton.edu}. Research supported in part by NSF grant  CFF-1319299.} \and
Noga Alon\thanks{Sackler School of Mathematics and Blavatnik School of Computer Science, Tel Aviv University,
Tel Aviv, Israel and School of Mathematics, Institute for Advanced Study, Princeton, USA, 
\texttt{nogaa@tau.ac.il}. Research supported in part by a USA-Israeli BSF grant,
by an ISF grant, by the Israeli I-Core program and by the Oswald
Veblen Fund.} \and 
Afonso S.\ Bandeira\thanks{Program in Applied and Computational Mathematics, Princeton University, Princeton, USA,
\texttt{ajsb@math.princeton.edu}. Research supported by AFOSR Grant No. FA9550-12-1-0317.}
\and 
Colin Sandon\thanks{Department of Mathematics, Princeton University, Princeton, USA,
\texttt{sandon@princeton.edu}.
}}
\date{}

\maketitle

\begin{abstract}
This paper considers the problem of linear Boolean classification, where the goal is to determine in which set, among two given sets of Boolean vectors, an unknown vector belongs to by making linear queries. Finding the least number of queries is equivalent to determining the minimal rank of a matrix over $GF(2)$ whose kernel does not intersect a given set $S$. In the case where $S$ is a Hamming ball, this reduces to finding linear codes of largest dimension. For a general set $S$, this is an instance of ``the critical problem'' posed by Crapo and Rota in 1970, open in general.
This work focuses on the case where $S$ is an annulus. As opposed to balls, it is shown that an optimal kernel is 
composed not only of dense but also of sparse vectors, and the optimal mixture is identified in various cases. 
These findings corroborate a proposed conjecture that for an annulus of inner and outer radius $nq$ and $np$ respectively, the optimal relative rank is given by the normalized entropy $(1-q)H(p/(1-q))$, an extension of the Gilbert-Varshamov bound. 
\end{abstract}

%
%

\section{Introduction}
We introduce the problem {\it linear Boolean classification} with the following example. Let $S$ be a data set of images partitioned into $k$ categories, i.e., $S = S_1 \sqcup S_2 \sqcup \ldots \sqcup S_k$, where each $S_i$ is a class of images (e.g., cats, apples, etc.). Think of each image $x$ as vector of pixels, say $x \in \F_q^n$, for some fixed $q$. We consider $q=2$ in what follows for simplicity (e.g., black and white images). Note that the use of a field $\F_q$ is irrelevant for now, any finite set would do. 

Let $x$ be an unknown image that belongs to $S$. One would like to determine in which class $x$ belongs to. One possibility is to query each entry of $x$ and then look for which class it belongs to. This would take $n$ queries in the worst-case, assuming that queries are not made adaptively. Is it possible to make fewer queries? Let us consider some examples. 

\begin{exa}
If $S$ contains two classes (i.e., $k=2$), the Boolean vectors having at most $s$ ones, and the Boolean vectors having at most $s$ zeros. \end{exa}
Then it is easy to check that $2s+1$ queries are enough to determine in which set an unknown vector belongs to: $2s$ queries may reveal $s$ ones and $s$ zeros, but an additional query would break ties. 
\begin{exa}
If $S$ contains two classes and the first class contains all Boolean vectors with exactly $s$ ones and the second class contains all Boolean vectors with exactly $s+2$ ones, with, say, $s =n/8$. 
\end{exa}
How many bits should one query to classify an unknown $x$? 
At this point, one should specify what is allowed for a ``query.'' If querying means simply to pole a subset of coordinates, then the problem (for any classes) is simply to look for the smallest subset of coordinates $S \subseteq [n]$ for which no two vectors in different classes can be confused. If the sets are permutation invariant like in the last two examples, this  amounts to look for the minimal number of components for which no two vectors in different classes can have equal weight. In the latter example, this requires poling $n-s +1 \sim \frac{7}{8}n$ coordinates. This is not a particularly interesting problem. 

Assume now that one is allowed to make ``linear queries'' and not just poling, that is, a query can ask for the parity of a subset of components. We refer to this as the {\it linear Boolean classification} problem. From an application point of view, this is motivated by the fact that linear measurements may be accessible without acquiring the full vector, like in compressed sensing. From a theory of computation point of view, this is assuming a model where XOR gates are free. Can one further improve on poling with linear queries? 

In the first example above, linear queries do not help, however they do in the second example. A set of linear queries allows to discriminate the two classes if there is no pair of elements $x,x'$ in the different classes that have the same output with the queries. If $A$ is an $m \times n$ matrix where each row represents a linear query, this means that
\begin{align}
&Ax \neq Ax', \quad \forall x \in S_1, x' \in S_2 \\
\Leftrightarrow \quad & A v \neq 0, \quad \forall v \in S_1 + S_2,
\end{align}
where $S_1 + S_2$ is the set of vectors obtained by adding vectors from each set. In the second example above, $S_1+S_2$ is the set of vectors having even weight at most $2s+2$, without the all-zero vector.  Hence $Av \neq 0$ for all $v$ in the sumset means that the linear code whose parity-check matrix is $A$ has no codeword of even weight at most $2s+2$. We know from coding theory that there exists codes of distance $d=2s+2$ with dimension \begin{align}
k=n(1-H(d/n)) + o(n),
\end{align} 
where $H$ is the entropy function (flattened at $1$ when the argument is more than $1/2$), and where the above can be achieved with a random code (the Gilbert-Varshamov bound). Hence the normalized number of queries $m/n$ can be close to $H(d/n) \sim H(1/4)$. Since $H(1/4) < \frac{7}{8}$, linear queries improve on poling. 

This gap can be made arbitrarily large by replacing $s=n/8$ with $s=\e n$, $\e \to 0$. The idea is that as the vectors in the classes get sparser, dense linear queries will be much more efficient in discriminating the two classes than poling, which is equivalent to projecting on unit vectors. Note also that the previous example is not a special case. For arbitrary classes, one expects that relying also on dense projections should outperform poling. Poling can be optimal, like in the first example above, when the subset $S_1+S_2$ is itself dense, in which case it may be preferable to project on sparse queries. A natural problem is hence to study these tradeoffs depending on the structure of $S_1$ and $S_2$, or more specifically, on $S_1+S_2$. 

\subsection{Results}
In this paper, we focus on the case of two symmetric classes. By symmetric, we mean that for any vector in the class, all its permutations are also in the class. Any such class can be described by providing a list of ``permitted'' Hamming weights. 
This paper further focuses on the case where the classes are annuli, i.e., where the permitted weights are in intervals $[u,v]$ with $0 \leq  u \leq v \leq n$. Besides for special cases, the sumsets will then be annuli of even weights, say in $[a,b]$. Hence, since $a$ may be positive, this differs from coding theory in the sense that we are allowed to have codewords of low enough weight, rather than allowing only for codewords of high enough weight (i.e., codes with specific weights).

In a first version of this paper \cite{isit_annulus}, the first three authors conjectured that the least number of linear queries needed to classify two sets whose sumset is the annulus $[a,b]$, $a<b$, satisfies the property 
\begin{align}
m^*(a,b,n)\stackrel{?}{=}m^*(1,b,n-a+1). \label{trans}
\end{align}
This would mean that the least number of linear queries for classifying sets whose sumset is the annulus $[a,b]$ in $\F_2^n$ is obtained from the annulus $[1,b]$ (which is a pinged ball centered at $0$) in $\F_2^{n-a+1}$.
In other words, the optimal tradeoff for classification would be reduced to tradeoffs of linear coding in smaller dimension. 
In terms of choosing the linear queries, equality \eqref{trans} reads as follows:
the optimal queries can taken as the vectors having zero's in the first $a-1$ components, and for the remaining $n-a+1$ components, the parity-check matrix of an optimal code (e.g., drawn uniformly at random) of length $n-a+1$ and distance $b$. 
Note the following two case:
\begin{itemize}
\item if $m^*(1,b,n-a+1)=n-a+1$, i.e., if the parity-check matrix of the code spans the entire space $\F_2^{n-a+1}$, then any set of $n-a+1$ unit vectors are optimal queries. This is for example the case in the first example of the introduction, where $S_1$ and $S_2$ are balls at the $0^n$ and $1^n$ respectively, and where poling is optimal. 
\item if $m^*(1,b,n-a+1)<n-a+1$, the parity-check matrix of the code can be chosen in a standard form with the identity matrix in the first $m^*(1,b,n-a+1)$ components and a `dense' matrix in the last $n-a+1-m^*(1,b,n-a+1)$ components. In the second example of the introduction,\footnote{Note that in this example, we ignored the fact that the vectors are allowed to have odd weight.} $a=1$ so that there is no zero-padding and the queries are obtained from the code. 
\end{itemize}
One can also interpret the optimal queries in terms of their orthogonal space: it consists of $a-1$ unit vectors and on the remaining components, a code of length $n-a+1$ and distance $b$.

This paper shows that \eqref{trans} holds for all $1 \leq a < b  \leq n$ with $b \geq 2a-2$ (Proposition \ref{mainprop}), but gives a counter-example to the general statement (Proposition \ref{counterex}). However, it remains a plausible conjecture that \eqref{trans} ``holds asymptotically,'' i.e., for any $0 \leq \alpha < \beta \leq 1$, 
\begin{align}
m^*(\alpha n,\beta n, n)=n(1-\alpha) \tilde{H}\left(\frac{\beta}{1-\alpha}\right) +o(n),
\end{align}
where $\tilde{H}(c)$ is the optimal compression rate of a code of relative distance $c$ (which matches the entropy function $H(c)$ if the GV bound is tight).  Along the way, we discuss connections to additive combinatorics (the Freiman-Ruzsa conjecture) for non-symmetric sets.

\subsection{Background on linear coding}\label{introa}
One of the fundamental problems of coding theory is to identify
the largest dimension of a binary code with a given 
length and
distance. This means to identify the largest cardinality of a subset
of $\F_2^n$ whose elements are at distance at least $d$ from each
other. This is a well-known open problem in general. Even for the case
of a {\it linear code}, i.e., a subspace of $\F_2^n$, the problem is
open. From the parity-check matrix viewpoint, constructing a linear code of
distance $d$ is equivalent to constructing a matrix $M$ such that $Mx$
allows to recover $x$ for all $x$ having weight at most $\left\lfloor (d-1)/2
\right\rfloor$, equivalently, to construct a matrix $M$ of least rank such that
\begin{align}
M x \neq Mx', \quad \forall x,x' \in B(0^n,s), x \neq x', \label{prop1}
\end{align}   
where 
\begin{align}
B(0^n,s)&=\{ x \in \F_2^n : w(x) \leq s\},\\
s&=\left\lfloor (d-1)/2 \right\rfloor,
\end{align}   
is the Hamming ball of centre $0^n$ and radius $s$, 
and where $w(x)$ denotes the Hamming weight of $x$. 
Note that for a fixed $s$ and $n$, the least rank of matrices satisfying property \eqref{prop1} is a finite integer, we denote it by $m^*(2s,n)$. The factor 2 will be justified later.  
Finding $m^*$ for general values is a difficult problem as mentioned previously, and even in the asymptotic regime of $s,n$ diverging with a fixed ratio $s/n=\delta/2$, the problem is still open. 
In fact, it is believed by some (e.g., Goppa's conjecture) that the answer is given by 
\begin{align}
m^*(\delta n,n)=nH(\delta)+o(n),
\end{align}   
the Gilbert-Varshamov (GV) bound \cite{gilbert-bound,varshamov-bound}, where 
\begin{align}
H(\delta)=  -\delta \log_2 \delta - (1-\delta) \log_2 (1-\delta),  
\end{align}
if $\delta \in [0,1/2]$ and $H(\delta)=1$ if $\delta \in (1/2,1]$. 

It is not difficult to establish this bound, with a greedy algorithm or with a
probabilistic argument, but it has not been improved since the 50's in
the asymptotic regime, nor has it been proved to be tight.

Linear codes are particularly interesting for several reasons. Their
encoding complexity is reduced from exponential (in the worst case) to
quadratic in the blocklength (specifying a basis). Moreover, most of the
codes studied in the literature and used in applications with efficient
decoding algorithms are linear. There are also other interesting features
of linear codes, such as their duality with linear source codes. The
parity-check matrix of a linear code can be viewed as a linear source
compressor, for a source distribution (or a source model in the worst-case
setting) equivalent to the error distribution (or model). In particular,
if the source model is given by the $k$-sparse sequences, i.e., binary
sequences with at most $k$ ones, then the optimal compression dimension,
assuming the Gilbert-Varshamov bound to be tight, is given by $nH(2k/n)
+ o(n)$, where the first term is approximately 
$2k\log n/k$,
when $k/n$ is small. 

\section{Preliminaries and ``the critical problem''}
\subsection{Linear coding for general models}
As expressed in \eqref{prop1}, linear coding can be viewed as constructing
flat matrices which are injective for sequences constrained to have a
bounded number of ones. This concerns linear coding for the traditional
Hamming ball model. One can consider more general models, in which case
the injectivity property \eqref{prop1} needs to be guaranteed for vectors
$x$ belonging to a specified set $S \subseteq \F_2^n$.  From now on,
we define linear codes by means of parity-check matrices.

\begin{Def}
A linear code $M: \F_2^n \to \F_2^m$ can compress losslessly a source model $S \subseteq \F_2^n$ (or can correct the error patterns in $S \subseteq \F_2^n$), if 
\begin{align}
M x \neq M x', \quad \forall x,x' \in S, x \neq x',
\end{align}
i.e., if $M$ is $S$-injective. 
\end{Def}

\begin{Def}
For a given set $S \subseteq \F_2^n$, we define the linear compression dimension of $S$ by 
\begin{align}
m^*(S)= \min_{M \in \text{$S$-injective}} \rank(M), 
\end{align}
and the linear compression rate of $S$ by $m^*(S)/n$.
\end{Def}
Note that $m^*(S)$ can be equivalently defined by 
\begin{align}
m^*(S)= n-\max_{V: (V \setminus \{0\}) \cap (S+S) = \emptyset} \dim(V), 
\end{align}
where $V$ denotes a subspace of $\F_2^n$. In other words, we need to find the largest dimension of a subspace avoiding $(S+S) \setminus \{0\}$. Throughout the paper, the sum of two sets is defined by $S_1+S_2=\{s_1+s_2: s_1 \in S_1, s_2 \in S_2\}$.  
The problem of finding the largest dimension of a vector space which does not intersect a given subset of $\F_q^n$ (where $q$ is a power of a prime) is known as ``the critical problem'' and was posed\footnote{In fact, an even more general formulation is proposed in \cite{crapo-rota}} by Crapo and Rota in \cite{crapo-rota}. 
Even for $q=2$, it is an open problem for arbitrary sets.  

Note also that if one is allowed to use a non-linear map $M$, required to be $S$-injective, the ``dimension'' of $M$ can be as low as $\log_2(|S|)$, but not lower. Quotes are used on ``dimension'' since the map is non linear and since a priori $\log_2(|S|)$ may not be an integer. For linear maps, we then have 
\begin{align}
\log_2(|S|) \leq m^*(S). \label{lower}
\end{align}
One can also obtain the following upper-bound with a probabilistic argument. 
\begin{lemma}\label{sum-bounds}
For any $S \subseteq \F_2^n$,
\begin{align}
\log_2(|S|) \leq  m^*(S) \leq \left\lfloor \log_2(|S+S|-1) \right\rfloor +1. \label{GV2}
\end{align}
\end{lemma}
A proof of this 
simple 
lemma is available in \cite{course1}. 
For Hamming balls, the above bounds are equivalent to the Hamming and Gilbert-Varshamov bounds. 
Note that if $S$ is a subspace, then the bounds match and are equal to $\log_2(|S|)$. 

\subsection{Linear coding for subspaces}
The lower bound \eqref{lower} is clearly achieved if $S$ is a subspace of $\F_2^n$, i.e., if $S+S=S$, using for $M$ the projection on $S$. 
One may ask for what kind of set $S$ are the two bounds in the lemma matching in the asymptotic regime, i.e., up to $o(n)$. This is  equivalent 
to asking for the doubling constant of $S$ to be sub-exponential, i.e., 
\begin{align}
 \frac{|S+S|}{|S|} = 2^{o(n)}. 
\end{align}
One may expect that this holds only if the set $S$ is closed to a subspace in some sense. In fact, this is related to the 
Polynomial Freiman-Ruzsa conjecture (see~\cite{Green_conjecture}):
\begin{conj}\label{ruzsa}[Polynomial Freiman-Ruzsa conjecture]
If $S$ has a doubling constant at most $K$, then $S$ is contained in the union of $K^{O(1)}$ translates of some subspaces of size at most $|S|$. 
\end{conj}
By
the probabilistic bound in \eqref{GV2}, we are therefore motivated to state the following conjecture.  
\begin{conj}\label{ruzsa-coding}
If the linear compression dimension of $S$ is given 
by $\log_2 (|S|) + o(n)$ (hence matches the non-linear compression dimension) then $S$ is contained in 
the union of $2^{o(n)}$  
translates of some subspaces of size at most $|S|$. 
\end{conj}
The above condition can only happen for sets which are non symmetric, i.e., not invariant under permutations of the $n$ coordinates (like a subspace), unless the set is very small (of size $2^{o(n)}$) or very large 
(of size $2^{n +o(n)}$). In the next section, we will focus on symmetric sets. 

\subsection{Linear Boolean classification for general models}
Given two disjoint classes $S_1,S_2 \subseteq \{0,1\}^n$, we are interested in constructing a linear map $M: \F_2^n \to \F_2^m$ such that for $x \in S_1 \cup S_2$, $Mx$ allows to determine if $x$ belongs to $S_1$ or $S_2$.
Our goal is to identify the least dimension $m$ for which linear classification is possible when the two classes are given. 

\begin{Def}
Let $S_1,S_2 \subseteq \{0,1\}^n$ with $S_1 \cap S_2 = \emptyset$. A linear classifier for $(S_1,S_2)$ is a linear map $M: \F_2^n \to \F_2^m$ which is $(S_1,S_2)$-separable in the sense that 
\begin{align}
M x_1 \neq M x_2, \quad \forall x_1 \in S_1,x_2 \in S_2.
\end{align}
The linear classification dimension of the set pair $(S_1,S_2)$ is defined by 
\begin{align}
m^*(S_1,S_2)= \min_{M \in \text{$(S_1,S_2)$-separable}} \rank(M).
\end{align}
\end{Def}
Note that $m^*(S_1,S_2)$ can also be expressed as 
\begin{align}
m^*(S_1,S_2)= n-\max_{V: V \cap (S_1+S_2) = \emptyset} \dim(V), 
\end{align}
where $V$ denotes a subspace of $\F_2^n$.  
In other words, we need to find the largest dimension of a subspace avoiding $S_1+S_2$. This is again an instance of ``the critical problem''. The linear classification rate is defined by $m^*(S_1,S_2)/n$.

Note however that the two simple bounds obtained in the previous section do not give interesting results here. First, if non-linear maps are allowed, the lower bound on $m^*$ is simply 1 as there are only two sets. In terms of probabilistic bounds, if $M$ is drawn uniformly at random, then we obtain 
\begin{align}
m^*(S_1,S_2) \leq \left\lceil \log_2(|S_1+S_2|) \right\rceil .
\end{align}
Unlike in the coding problem, this is not a strong bound in general, as shown in Section \ref{results}. In particular, it does not take advantage of the fact that $S_1+S_2$ is away from $0^n$. 

\subsection{Formal definitions and equivalences}
Let $n \in \mZ_+$ and $S,S_1,S_2 \subseteq \{0,1\}^n$. 
\begin{Def}
A matrix $M$ is 
\begin{itemize}
\item $S$-distinguishable if $Mx\neq 0$ for all $x\in S$,
\item $S$-injective if $Mx\neq Mx'$ for all $x,x'\in S$, $x \neq x'$,
\item $(S_1,S_2)$-separable if $Mx_1\neq Mx_2$ for all $x_1 \in S_1,x_2\in S_2$.
\end{itemize}
\end{Def}
\noindent
Note that
\begin{itemize}
\item $M$ is $S$-distinguishable iff $\ker(M) \cap S = \emptyset$,
\item $M$ is $S$-injective iff $M$ is $((S+S) \setminus \{0\})$-distinguishable, 
\item $M$ is $(S_1,S_2)$-separable iff $M$ is $(S_1+S_2)$-distinguishable.
\end{itemize}
\noindent
From the first item above, a matrix of minimal rank which is $S$-distinguishable can be equivalently constructed by finding a space of maximal dimension which avoids $S$, an instance of ``the critical problem''. 
The meaning of previous mathematical objects in terms of compression, coding and classification notions are:
\begin{itemize}
\item $M$ is a lossless compressor for the source model $S$ iff $M$ is $S$-injective,
\item $M$ is the parity-check matrix of an error-correcting code for the error model $S$ iff it is a lossless compressor for the source model $S$,
\item $M$ is a linear classifier for the class models $(S_1,S_2)$ iff $M$ is $(S_1,S_2)$-separable.
\end{itemize}
One could also use classification in a coding context for errors, 
to determine if an error pattern belongs to two 
different classes (e.g., typical or atypical patterns, high or low 
SNR regimes). In this case the linear 
constraint
on the classifier is important to allow the source-channel coding duality.

In this paper, we are interested in symmetric sets, i.e., sets which are
invariant under permutations.  Note that these are defined by the Hamming
weights of their elements.  One of the most fundamental symmetric sets is
the Hamming ball at $0^n$, studied extensively in coding theory.  The only
other symmetric Hamming ball is the one centered at $1^n$.  A natural
set structure to consider next is the annulus, which contains the cases
above and extends to more general symmetric sets.  In particular, an
arbitrary symmetric set is a union of annuli. In addition, the sum of
two annuli, which matters for classification, is again an annulus
(except in particular cases where the two annuli have a single weight,
and their sum contains only even or odd weight vectors).

\begin{Def}
Let $1 \leq a \leq b \leq n$, 
\begin{align}
A(a,b,n)&=\{x \in  \{0,1\}^n : w(x) \in [a,b]\},\\
m^*(a,b,n)&=\min_{M \in A(a,b,n)\text{-distinguishable}} \rank(M)\\
&=n-\max_{V : V \cap A(a,b,n)=\emptyset } \dim(V),
\end{align}
where $M$ is a matrix over $\F_2$ with $n$ columns and $V$ is a subspace of $\F_2^n$.
\end{Def}
Note that $m^*(1,b,n)=m^*(b,n)$, i.e., when the annulus degenerates 
into the punctured ball $B(0^n,b) \setminus \{0^n\}$, 
the definition is consistent with the one of Section \ref{introa}. 

Our goal is to characterize $m^*(a,b,n)$ for finite values of the parameters and in the asymptotic regime. 


\section{Results}\label{results}
Recall that $m^*(1,b,n)$, the least rank of a parity-check matrix of a code of distance $b+1$ on a blocklength $n$, has no known explicit form in general. It is clear that 
$m^*(1,b,n)=n$, if $b = n$, 
and it is known that 
\begin{align}
m^*(1,b,n)=n + o(n), \quad \text{if } b \geq n/2, \label{plotkin}
\end{align}
from Plotkin's bound \cite{plotkin}. Also, the GV bound provides the inequality 
\begin{align}
m^*(1,\beta n,n) \leq nH(\beta) + o(n),
\end{align}
which is conjectured to be tight. In this paper we are interested in characterizing $m^*(a,b,n)$ for $a \geq 2$, possibly in terms of $m^*(1,b,n)$. In other words, we want to study how the ``hole'' in an annulus allows one to decrease the dimension of $M$ when compared to a Hamming ball. 

We obtain the following result.
\begin{prop}\label{mainprop}
For any $1 \leq a \leq b  \leq n$ with $b \geq 2a-2$, 
\begin{align}
m^*(a,b,n)=m^*(1,b,n-a+1).
\end{align}
\end{prop}
Assuming the GV bound to be tight, the above result takes the following asymptotic form
\begin{align}
m^*(\alpha n,\beta n, n)=n(1-\alpha) H\left(\frac{\beta}{1-\alpha}\right) +o(n).
\end{align}
The proof below shows that an optimal kernel for an $A(a,b,n)$-distinguishable matrix is composed by a mixture of $a-1$ unit vectors and a code of largest dimension and minimal distance at least $b+1$ on the remaining components.  
\begin{proof}
We work with the kernel approach, maximizing the dimension of $V=\ker(M)$. We use the notation $k(a,b,n)=n-m^*(a,b,n)$ and provide a proof of Proposition \ref{mainprop} based on the weight of the sparsest vector in an optimal basis.

\begin{claim}\label{mainlemma}
If $b\geq 2a-2$ and $a\geq 2$,
\[
k(a,b,n) \leq \max_{1\leq  s \leq a-1} \left[ s + k(a-s,b,n-s) \right].
\]
\end{claim}
\begin{proof}[Proof of Claim \ref{mainlemma}]

Let $V$ be a subspace that does not intersect $A(a,b,n)$, 
and $s$ be the sparsity of the sparsest non-zero vector in $V$. If $s\geq a$ then $\dim(V) \leq k(1,b,n)$ and using $k(1,b,n-1) +1 \geq k(1,b,n)$,  
\begin{align}
\dim(V) &\leq k(1,b,n) \leq 1+ k(1,b,n-1) \\ &\leq 1+ k(a-s,b,n-1) \\
 &\leq \max_{1\leq s \leq a-1} \left[ s + k(a-s,b,n-s) \right].
\end{align}
On the other hand, if $s < a$, we will show that $\dim(V) \leq s + k(a-s,b,n-s)$, proving the Lemma.

Let $v$ be a vector in $V$ that is exactly $s$-sparse. We permute (the coordinates of) $V$ so that $v$ is 1 in the first $s$ components and $0$ elsewhere. We represent $V$ by a matrix (below) with its rows forming a basis of $V$, we pick such a representation such that $v$ is the first row.
\[
\left[\begin{array}{cc} 1_{1\times s} &  0_{1\times (n-s)} \\  \vdots & \vdots \end{array} \right],
\]


If $\dim(V) < s$ then the result is trivial, so we will focus on the case $\dim(V) \geq s$, using Gauss Elimination one can find a basis of $V$ such that $v$ is the first row and the first $s$ by $s$ block is upper triangular, meaning:

\[
\left[\begin{array}{ccc}
1 & 1_{1\times (s-1)} &  0_{1\times (n-s)} \\
0_{(s-1)\times 1} & T & R \\
0_{(\dim(V)-s)\times 1} & 0_{(\dim(V)-s)\times (s-1)} & V^{\ast}
\end{array} \right],
\]
with $T$ upper-triangular.

Note that $V^\ast$ is a basis for a subspace in $n-s$ coordinates of dimension $\dim({ V})-s$. We next argue that $V^\ast \cap A(a-s, b, n-s) = \emptyset$. Indeed, suppose  $u \in V^\ast \cap A(a-s, b, n-s)$ then

\begin{itemize}

\item If $a \leq w(u) \leq b$ then the vector $[0_{s\times 1} \ u]\in V$ is in $A(a,b,n)$.

\item On the other hand, if $a-s \leq w(u) < a$ the vector $ [0_{s\times 1} \ u]\in V$ summed to the $s$ sparse vector $v\in V$ will give a vector in $A(a,a-1+s,n)$. Since $a-1+s \leq a-1+a-1 = 2a-2 \leq b$ then $A(a,a-1+s,n) \subset A(a,b,n)$.

\end{itemize}

This means that $\dim(V^\ast) \leq k(a-s,b,n-s)$ and  
$\dim(V) - s \leq k(a-s,b,n-s)$. 
\end{proof}

The proof of Proposition \ref{mainprop} follows then by a strong induction on $a$ and Proposition~\ref{propUB}.
\end{proof}

\begin{remark}
Note that the proof above would carry through if $b < 2a-2$, as long as there exists an element in the optimal subspace, whose Hamming weight is smaller or equal to $b-a+1$. This means that, for Conjecture~\ref{mainconjecture}, the subspace must have all its elements' weights not only avoiding $[a,b]$ but also avoiding $[1,b-a+1]$.
\end{remark}

\begin{remark}
 An intuitive way of thinking about the condition $b\geq 2a-2$ is that it enforces that any sum of sparse vectors $x,y \in A(a,b,n)$ (meaning $w(x),w(y) \leq a-1$) cannot be dense, as $w(x+y) \leq  b$. Indeed, one can use this fact to provide an alternative proof to Proposition~\ref{mainprop}.
\end{remark}

Using \eqref{plotkin}, we obtain the following corollary which provides a characterization of $m^*$ in certain cases. 
\begin{corol}\label{maincorol}
If $b \geq 2a-2$ and $2b\geq n-a+1$, 
\begin{align}
m^*(a,b,n)=n-a+1 +o(n).
\end{align}
\end{corol}
Note that in terms of the classification problem, the above applies to the cases of two sets having a sumset which is dense, such as two Hamming balls $0^n$ and $1^n$, in which case poling of the components is optimal and one does not have to rely on a dense code.  

It is also straightforward to show that the following holds for any $a$:
\begin{align}\label{maincorolEASYfollows}
m^*(a,n,n)=n-a+1.
\end{align}
The proof simply uses the fact that a subspace containing $a$ 
linearly independent vectors must generate a vector of weight 
at least 
$a$. 
 
Going back to the example where
$S_1=B(1^n,s)$ and $S_2=B(0^n,s)$, $s < n/2$, are the two classes,
we have $S_1+S_2=B(1^n,2s)=A(n-2s,n,n)$. Hence, by (\ref{maincorolEASYfollows}),
$m^*=2s+1$. This is indeed verified by the fact that the identity matrix
of dimension $2s+1$ allows to classify the two sets (by a simple majority
count on any $2s+1$ coordinates). This also shows that the classification
dimension can be simpler to find than the compression dimension, since
the sum of two disjoint sets $S_1+S_2$ is typically away from $0^n$ and
hence $V$ can contain sparse vectors, whereas the sumset $S+S$ contains a
ball around $0^n$ and requires a packing of dense vectors, which is more challenging.

Note that conjecture Proposition \ref{mainprop} does not hold for any values of $a,b,n$. For example, the following case is a counter-example, 
\begin{align}\label{degeneratecaseform}
m^*(a,a,n)=1, \quad \text{$a$ odd, $n$ even}.
\end{align}
It is straightforward to verify (\ref{degeneratecaseform}) by noting that packing vectors of even weight will never produce an odd weight vector.
Moreover, as shown in Section \ref{counter}, the previous ``pathological'' case is not the only exception to the conclusion of Proposition \ref{mainprop}. However, we conjecture that the asymptotic version of Proposition 1 holds universally.  
\begin{conj}\label{mainconjecture}
For any $0 \leq \alpha < \beta \leq 1$, 
\begin{align}
m^*(\alpha n,\beta n, n)=n(1-\alpha) \tilde{H}\left(\frac{\beta}{1-\alpha}\right) +o(n),
\end{align}
where $\tilde{H}(c)$ is the optimal compression rate of a code of relative distance $c$ (which matches the entropy function $H(c)$ if the GV bound is tight).  
\end{conj}

The upper-bound holds in full generality.  
\begin{prop}\label{propUB}
For any $1 \leq a \leq b \leq n$, 
\begin{align}
m^*(a,b,n) \leq m^*(1,b,n-a+1).
\end{align}
\end{prop}
To achieve the above, it is enough to take for the kernel of $M$ a basis consisting of $a-1$ unit vectors and an optimal subspace of weight at least $b+1$ on the complement coordinates.  


We further obtain a few more cases which corroborate the conjecture. 

In particular, as observed by the second author in the 80s (c.f. \cite{enomoto}),
a theorem of Olson \cite{Olson_abelian} 
can be used to show that if $a$ is a power of
$2$ then any
binary linear code of dimension  $a$ (and any length) contains a
vector of
Hamming weight divisible by $a$. This implies that if $a>n/2$ is a
power of
$2$ then 
$k(a,a,n)=n-m^*(a,a,n)=a-1$. A more general
result is proved in \cite{enomoto} where it is shown that
for any even $a>n/2$, $k(a,a,n)=a-1$. Further results on codes with a forbidden distance can be found in \cite{cohen}. 

\section{A counter-example to the generality of Proposition \ref{mainprop}}\label{counter}
\begin{prop}\label{counterex}
There exists $1 \leq a<b \leq n$ such that $m^*(a,b,n)\ne m^*(1,b,n-a+1)$.
\end{prop}
This disproves Conjecture 3 in \cite{isit_annulus}. 

\begin{proof}
For any $n$ and $d$, let 
\begin{align}
V_{n,d}=\{x\in \mathbb{F}_2^n: w(x)\equiv 0\pmod{2d} \text{ and } \forall i\le \left\lfloor \frac{n}{d}\right\rfloor, x_{di-d+1}=x_{di-d+2}=...=x_{di}\}.
\end{align}
We will show that $V_{n,2}$ already provides a counter-example, but carry the more general argument first. 
Note that $V_{n,d}$ is a $(\left\lfloor \frac{n}{d}\right\rfloor-1)$-dimensional subspace of $\mathbb{F}_2^n$. Also, for any $a$ and $b$ such that there is no multiple of $2d$ in $[a,b]$, $V_{n,d} \cap A(a,b,n)=\emptyset$. 

Next, let 
\begin{align}
W_{n,d}=\{x\in \mathbb{F}_2^n: \forall i\le d\left\lfloor \frac{n}{d}\right\rfloor \text{ s.t. } d\left\lceil \frac{i}{d}\right\rceil-i< \frac{d}{2}, x_i=0 \}.
\end{align}
In other words, $W_{n,d}$ is the subspace of $\mathbb{F}_2^n$ consisting of all vectors such that when the first $d\left\lfloor \frac{n}{d}\right\rfloor$ indices are divided into blocks of length $d$, every component of the vector that corresponds to an index in the second half of its block is $0$. Every element of $V_{n,d}$ has all components with indices in the same block set to the same value, so for any $x\in V_{n,d}$ and $y\in W_{n,d}$, $w(x+y)\ge w(y)$.

Now, for any $n$, $d$, and $b$, let $W'_{n,d,b}$ be the highest dimensional subspace of $W_{n,d}$ with no nonzero element of Hamming weight $b$ or less. $W_{n,d}$ has dimension
\begin{align}
\left\lfloor \frac{d}{2}\right\rfloor\cdot \left\lfloor \frac{n}{d}\right\rfloor +n-d\left\lfloor \frac{n}{d}\right\rfloor,
\end{align}
so $W'_{n,d,b}$ has dimension
\begin{align}
 \left\lfloor \frac{d}{2}\right\rfloor\cdot \left\lfloor \frac{n}{d}\right\rfloor +n-d\left\lfloor \frac{n}{d}\right\rfloor-m^*(1,b,\left\lfloor \frac{d}{2}\right\rfloor\cdot \left\lfloor \frac{n}{d}\right\rfloor +n-d\left\lfloor \frac{n}{d}\right\rfloor).
\end{align} 
Also, for any $x\in V_{n,d}$ and $y\in W'_{n,d,b}$, if $y\ne 0$ then $w(x+y)\ge w(y)>b$ and if $y=0$, $w(x+y)$ is divisible by $2d$. So, given any $a\le b$ such that there is no multiple of $2d$ in $[a,b]$, it follows that 
\begin{align}
 m^*(a,b,n)\le \left\lceil \frac{d-2}{2}\right\rceil\cdot \left\lfloor \frac{n}{d}\right\rfloor+1+m^*(1,b,\left\lfloor \frac{d}{2}\right\rfloor\cdot \left\lfloor \frac{n}{d}\right\rfloor +n-d\left\lfloor \frac{n}{d}\right\rfloor).
\end{align}
In many cases, it is unclear how this compares to $m^*(1,b,n-a+1)$. However, if $d=2$, this implies that 
\begin{align}
m^*(a,b,n)\le 1+m^*(1,b,\left\lceil \frac{n}{2}\right\rceil)\le \left\lceil \frac{n}{2}\right\rceil+1.
\end{align}
For every $\delta>\frac{1}{3}$, there exists $c$ such that for all $\delta n\le a\le b$, $m^*(b,n-a+1)\ge n-a-c$. So, for all sufficiently large $n$ and all $\delta n\le a\le b<\frac{n}{2}-c-1$ such that there is no multiple of $4$ in $[a,b]$, $m^*(1,b,n-a+1) \geq n-a-O(1)$, which can be made close to $2/3n$, hence much larger than $ \left\lceil \frac{n}{2}\right\rceil+1\ge m^*(a,b,n)$. Therefore, 
\begin{align}
m^*(a,b,n)\ne m^*(1,b,n-a+1).
\end{align}
\end{proof}
Note however that this does not seem to give a counter-example to Conjecture \ref{mainconjecture}.

\section{Future work}
For symmetric sets, the next steps would be to investigate further regimes 
for the radius of the annulus, supporting or disproving the
conjecture or to consider the union of two annuli. 
It would also be interesting to consider a probabilistic
rather than worst-case model for the linear Boolean classification
problem, as well as an adaptive query model.

The complexity of the classification would be another
interesting direction to pursue. In this paper, we do not specify algorithms to decide among the two classes by
accessing the linear queries $y=Mx$. A general approach is 
to
find a solution $x_0$ of $y=Mx$,
and to then verify if $x_0 + \ker (M)$ intersects $S_1$ or $S_2$. This
may be of course computationally costly, but our goal in this paper is
only to identify the least rank of $M$ for which the intersection always
happens only with one of the two sets. It would be interesting to consider computational efficient
classifiers.

Finally, a natural extension is to
consider the problem of constructing matrices that allow to classify
certain sets while compressing others. This may be used to construct codes that allow to narrow down the search of the transmitted codeword by successively discarding subsets of possible codewords. This will be other instances of
``the critical problem''.

\section*{Acknowledgement}
E.~Abbe was supported in part by NSF grant CIF-1706648. N.~Along was supported in part by a USA-Israeli BSF grant,
by an ISF grant, by the Israeli I-Core program and by the Oswald Veblen Fund. A.~S.~Bandeira was supported by AFOSR Grant No. FA9550-12-1-0317.

%
%


\begin{thebibliography}{1}

\bibitem{isit_annulus}
E.~Abbe, N.~Alon, A.~Bandeira, ``Linear Boolean classification, coding, and ``the critical problem'','' in Proc. ISIT, Hawaii, 2014. 


\bibitem{course1}
E.~Abbe, \emph{Worst-case source coding}, Course notes: Coding Theory and
  Random Graphs, Princeton University. Available at www.princeton.edu/eabbe,
  2013.

\bibitem{cohen}
L.~Bassalygo, G.~Cohen, and G.~Z\'emor, \emph{Codes with forbidden distances,} Discrete Mathematics, Volume 213, Issues 1--3, Pages 3--11, February 2000.

\bibitem{crapo-rota}
H.H. Crapo and G.C. Rota, \emph{On the foundations of combinatorial theory:
  Combinatorial geometries}, MIT Press, Cambridge, MA, 1970.

\bibitem{enomoto}
H.~Enomoto, P.~Frankl, N.~Ito, and K.~Nomura, \emph{Codes with given
  distances}, Graphs and Combinatorics \textbf{3} (1987), no.~1, 25--38.

\bibitem{gilbert-bound}
E.N. Gilbert, \emph{A comparison of signalling alphabets}, Bell System
  Technical Journal \textbf{31} (1952), no.~3, 504--522.

\bibitem{Green_conjecture}
B.~Green, \emph{Notes on the polynomial {F}reiman-{R}uzsa conjecture},
  available online (2005), http://people.maths.ox.ac.uk/greenbj/papers/PFR.pdf.

\bibitem{Olson_abelian}
J.E. Olson, \emph{A combinatorial problem on finite abelian groups, i}, Journal
  of Number Theory \textbf{1} (1969), no.~1, 8--10.

\bibitem{plotkin}
M.~Plotkin, \emph{Binary codes with specified minimum distance}, Information
  Theory, IRE Transactions on \textbf{6} (1960), no.~4, 445--450.

\bibitem{varshamov-bound}
R.R.. Varshamov, \emph{Estimate of the number of signals in error correcting
  codes}, Dokl. Acad. Nauk SSSR \textbf{117} (1957), 739--741.

\end{thebibliography}

\end{document}